\documentclass[11pt]{article}

\usepackage[margin=1in]{geometry}
\usepackage{amsmath,amssymb,amsthm,mathtools}
\usepackage{microtype}
\usepackage{enumitem}
\usepackage{graphicx}
\usepackage{float}
\usepackage[section]{placeins}
\usepackage{xcolor}
\usepackage{tikz}
\usepackage{pgfplots}
\usetikzlibrary{arrows.meta,calc,decorations.pathreplacing,positioning,fit,backgrounds,shapes.geometric,patterns,shadows,shapes.misc}

\definecolor{ptblue}{HTML}{1E5AA8}
\definecolor{ptlightblue}{HTML}{EAF2FB}
\definecolor{ptorange}{HTML}{D97A1D}
\definecolor{ptlightorange}{HTML}{FFF1E3}
\definecolor{ptgreen}{HTML}{237A4B}
\definecolor{ptlightgreen}{HTML}{EAF7EE}
\definecolor{ptpurple}{HTML}{6A4C93}
\definecolor{ptlightpurple}{HTML}{F1EAF8}
\definecolor{ptgray}{HTML}{5F6B7A}
\definecolor{ptdark}{HTML}{263445}
\definecolor{ptred}{HTML}{B63838}
\definecolor{ptlightred}{HTML}{FBEAEA}
\definecolor{ptgold}{HTML}{B8860B}
\definecolor{ptlightgold}{HTML}{FFF7D6}

\usepackage[colorlinks=true,linkcolor=ptblue,citecolor=ptblue,urlcolor=ptblue]{hyperref}

\pgfplotsset{compat=1.18}
\pgfmathdeclarefunction{exactpt}{1}{%
  \pgfmathparse{0.5*(1 + floor(1/(#1))*(#1)^2 + (1-floor(1/(#1))*(#1))^2)}%
}

\numberwithin{equation}{section}

\newtheorem{theorem}{Theorem}[section]
\newtheorem{lemma}[theorem]{Lemma}
\newtheorem{proposition}[theorem]{Proposition}
\newtheorem{corollary}[theorem]{Corollary}
\newtheorem{definition}[theorem]{Definition}
\newtheorem{remark}[theorem]{Remark}

\DeclareMathOperator{\Tr}{Tr}

\newcommand{\C}{\mathbb C}
\newcommand{\Id}{\mathrm{Id}}
\newcommand{\ket}[1]{\lvert #1\rangle}
\newcommand{\bra}[1]{\langle #1\rvert}
\newcommand{\braket}[2]{\langle #1\mid #2\rangle}
\newcommand{\abs}[1]{\left\lvert #1\right\rvert}
\newcommand{\norm}[1]{\left\lVert #1\right\rVert}
\newcommand{\set}[1]{\left\{#1\right\}}
\newcommand{\Prod}{\mathsf{Prod}}
\newcommand{\ov}{\operatorname{Overlap}}
\newcommand{\PT}{\operatorname{PT}}
\newcommand{\QMA}{\mathsf{QMA}}

\newcommand{\SWAP}{\mathsf{SWAP}}

\newcommand{\eps}{\varepsilon}
\usepackage{comment}

\hypersetup{
  pdftitle={An Optimal Analysis of the Product Test},
  pdfauthor={},
  pdfsubject={Quantum property testing},
  pdfkeywords={product test, quantum property testing, multipartite entanglement, QMA(2), product states, soundness}
}

\title{An Optimal Analysis of the Product Test}
\author{Jacob Beckey\thanks{University of Illinois Urbana--Champaign.
Email: \url{jbeckey@illinois.edu}} \and  Fernando Granha Jeronimo\thanks{University of Illinois Urbana--Champaign.
Email: \url{granha@illinois.edu}} \and Pei Wu\thanks{The Pennsylvania State University.
Email: \url{pei.wu@psu.edu}}}
\date{\today}

\begin{document}
\maketitle

\begin{abstract}
Product testing, i.e., deciding whether a pure multipartite quantum state is fully unentangled across a specified tensor decomposition, serves as a bridge between quantum property testing, unentangled quantum proof systems, and tensor optimization. The rigorous formalization and analysis of this as a quantum property testing problem was initiated by Harrow and Montanaro [FOCS 2010, J.ACM, 13], who analyzed a simple circuit that performs a SWAP test on each pair of corresponding registers, has perfect completeness, and gives guarantees that do not depend on the number of parties or on the local dimensions. Despite being a fundamental property testing primitive and having many applications, the product test's exact (worst-case) acceptance probability curve has yet to be fully determined.

In this work, we determine this curve exactly.  Let \(\omega\) be the maximum squared overlap of the input with a product state, and let \(\PT_n(\omega)\) be the largest possible acceptance probability of the product test over all \(n\)-partite pure states with product overlap \(\omega\), allowing arbitrary finite local dimensions.  Define
\[
  s(\omega)=\max\left\{\sum_j \lambda_j^2:
  \lambda_j\ge0,
  \sum_j \lambda_j=1,
  \lambda_j\le\omega\ \text{for all }j\right\}.
\]
If \(m=\lfloor1/\omega\rfloor\) and \(r=1-m\omega\), then \(s(\omega)=m\omega^2+r^2\). We prove that, for every \(n\ge2\) and every \(\omega\in(0,1]\),
\[
  \PT_n(\omega)=\frac{1+s(\omega)}2
  =\frac12\left(1+m\omega^2+(1-m\omega)^2\right).
\]
Thus on \(\omega\in(1/(m+1),1/m]\), the exact curve is the quadratic
\[
  \PT_n(\omega)=1-m\omega+\frac{m(m+1)}2\omega^2.
\]
The formula recovers the previously known tight branch \(1-\omega+\omega^2\) for \(\omega\ge1/2\), resolves all low-overlap regimes \(\omega<1/2\), and implies \(\PT_n(\omega)\to1/2\) as \(\omega\to0\) answering an open problem in [Soleimanifar and Wright, SODA 2022].  Lovitz and Lowe recently determined this curve for bipartite states ($n=2$). Our main result proves that the bipartite curve remains extremal for every number of parties $n$, thereby answering their open question.
As a complexity-theoretic application, the one-shot soundness parameter in the Harrow-Montanaro reduction from $\QMA(k)$ to $\QMA(2)$ improves from $1-(1-\sigma)^2/100$ to $1-(1-\sigma)^2/4$, where $\sigma$ is the original soundness. Our techniques, built upon those of Soleimanifar and Wright, allow us to resolve these open questions while remaining surprisingly elementary.
\end{abstract}
\clearpage

\section{Introduction}


Entanglement is an essential feature separating quantum information from its classical analogue. Thus, in many areas of quantum information, determining the entanglement content of a quantum state is of central importance. In algorithmic and complexity-theoretic settings, however, one often needs to certify the opposite property: that a state factors across a prescribed list of registers.  A pure state on $\mathcal H_1\otimes\cdots\otimes\mathcal H_n$ is fully product if it has the form
\[
  \ket{v_1}\otimes\cdots\otimes\ket{v_n}.
\]
This is the quantum analogue of complete independence across coordinates. In the setting of quantum complexity, product states represent the promise that multiple quantum witnesses are unentangled~\cite{harrow2010Efficient}. In mean-field theories, which are ubiquitous in physics and chemistry, the underlying systems are modeled as product states~\cite{brandao2016ProductState,bakshi2025Learning}. Product states are also often studied as a special case of matrix product states, which generalize mean-field theories by allowing for some degree of entanglement~\cite{soleimanifar2022Testing}. Finally, as pointed out in Ref.~\cite{harrow2010Efficient}, the product test has surprising applications to classical tensor optimization problems.

As such, testing productness is a fundamental task. Among all possible product testers, what properties are desirable? For many of the above applications, we require a procedure that does not depend on the dimension of the input states or the number of parties, utilizes only few-copy measurements, respects the tensor decomposition, and rejects states that are far from every product state. The circuit first introduced in Refs.~\cite{brennen2003observable,mintert2005Concurrence} and then rigorously studied in the property testing framework by Harrow and Montanaro~\cite{harrow2010Efficient} is the canonical procedure with these features. As such, we will refer to it as \textit{the} product test herein.

\subsection{From the SWAP test to the product test}

The historical path to the product test begins with the SWAP test~\cite{barenco1997Stabilization}, the elementary two-copy primitive used, for example, in quantum fingerprinting \cite{buhrman2001Quantum} and non-linear functional estimation~\cite{ekert2002Direct}. The SWAP test checks whether two registers are symmetric under exchange; for pure inputs this reduces to checking equality. Given a pure $n$-partite input state, the product test  (shown in Fig.~\ref{fig:product-test-schematic}) applies the SWAP test locally, register by register. As such, a product state passes every comparison with certainty, because each local pair is in a pure symmetric state. For an entangled state, however, the reduced state on a single pair of corresponding registers -- across two independent copies -- can have weight in the antisymmetric subspace, and the local test detects exactly this failure. Because of this feature, the parallelized SWAP-test circuit has been used in entanglement characterization for decades~\cite{brennen2003observable,mintert2005Concurrence,beckey2021Computable}.

More formally, given two copies of a pure state \(\ket\psi\in\mathcal H_1\otimes\cdots\otimes\mathcal H_n\), the product test performs an independent SWAP test on each pair of corresponding registers and accepts iff every local test accepts.  Equivalently, its accepting projector is
\begin{equation}\label{eq:product-projector-intro}
  \Pi_{\Prod,n}
  =\bigotimes_{i=1}^n \Pi_{\mathcal H_i},
  \qquad
  \Pi_{\mathcal H_i}=\frac{\mathbb{I}+\SWAP_i}{2}.
\end{equation}

\begin{figure}[t]
\centering
\resizebox{0.95\linewidth}{!}{%
\begin{tikzpicture}[
  >=Stealth,
  line cap=round,
  line join=round,
  every node/.style={font=\small},
  panel/.style={draw=ptgray!26, fill=ptgray!3, rounded corners=6pt, line width=.75pt},
  lane/.style={draw=ptblue!28, fill=white, rounded corners=4pt, line width=.60pt},
  reg/.style={draw=ptblue!82!black, fill=ptlightblue, rounded corners=2pt,
    minimum width=.96cm, minimum height=.45cm, line width=.75pt,
    font=\footnotesize\bfseries, text=ptdark},
  test/.style={draw=ptorange!90!black, fill=ptlightorange, rounded corners=3pt,
    minimum width=1.10cm, minimum height=.73cm, line width=.92pt,
    align=center, font=\footnotesize\bfseries, text=ptdark},
  accept/.style={draw=ptgreen!78!black, fill=ptlightgreen, rounded corners=4pt,
    align=center, text width=2.26cm, inner sep=6pt, line width=.90pt,
    font=\footnotesize\bfseries, text=ptdark},
  eqbox/.style={draw=ptgray!34, fill=white, rounded corners=4pt,
    align=center, inner sep=7pt, line width=.65pt, font=\scriptsize},
  labelbox/.style={draw=ptgray!34, fill=white, rounded corners=4pt,
    inner sep=4pt, line width=.65pt,
    font=\footnotesize\bfseries, text=ptorange!90!black},
  note/.style={font=\footnotesize\bfseries},
  wire/.style={draw=ptgray!70, line width=.76pt},
  arrow/.style={-{Stealth[length=2.0mm]}, draw=ptgray!72, line width=.86pt},
  bus/.style={draw=ptgray!55, line width=.70pt}
]
\draw[panel] (-.82,-3.34) rectangle (12.08,2.84);
\draw[lane] (.20,1.42) rectangle (7.93,2.08);
\draw[lane] (.20,-1.58) rectangle (7.93,-.92);
\node[anchor=east, note, text=ptblue!85!black] at (.04,1.75) {$\ket{\psi}$};
\node[anchor=east, note, text=ptblue!85!black] at (.04,-1.25) {$\ket{\psi}$};
\foreach \x/\lab in {1.05/1,2.62/2,4.19/3,7.24/n} {
  \node[reg] (u\lab) at (\x,1.75) {$\mathcal H_{\lab}$};
  \node[reg] (l\lab) at (\x,-1.25) {$\mathcal H_{\lab}$};
  \node[test] (s\lab) at (\x,0) {SWAP \\[-.5mm]test};
  \draw[wire] (u\lab.south) -- (s\lab.north);
  \draw[wire] (l\lab.north) -- (s\lab.south);
}
\node[font=\Large, text=ptgray] at (5.48,1.75) {$\cdots$};
\node[font=\Large, text=ptgray] at (5.48,0) {$\cdots$};
\node[font=\Large, text=ptgray] at (5.48,-1.25) {$\cdots$};
\draw[bus] (7.76,.28) -- (8.42,.28) -- (8.42,0);
\draw[bus] (7.76,-.28) -- (8.42,-.28) -- (8.42,0);
\node[accept] (acc) at (10.45,0) {accept iff\\[-.3mm]every SWAP\\[-.3mm]test accepts};
\draw[arrow] (8.42,0) -- (acc.west);
\draw[decorate, decoration={brace, amplitude=4.4pt}, draw=ptorange!90!black, line width=.72pt]
  (.70,.5) -- (7.60,.50)
  node[midway, yshift=.5cm, labelbox]
  {local symmetric projections};
\node[eqbox, text width=5.75cm] at (5.78,-2.68)
  {$\displaystyle \Pi_{\Prod,n}=\bigotimes_{i=1}^n\Pi_{\mathcal H_i}$\qquad
   $\displaystyle\Pi_{\mathcal H_i}=\frac{\mathbb{I}+\SWAP_i}{2}$\\[.5mm]
   {\scriptsize perfect completeness on product states}};
\end{tikzpicture}%
}
\caption{\textbf{The product test.} Given two copies of the same multipartite state, the product test performs a SWAP test on each pair of corresponding registers, and accepts only if all individual tests accept. Although the measurements act independently on each subsystem pair, the overall acceptance probability detects the global failure of product structure.}
\label{fig:product-test-schematic}
\end{figure}
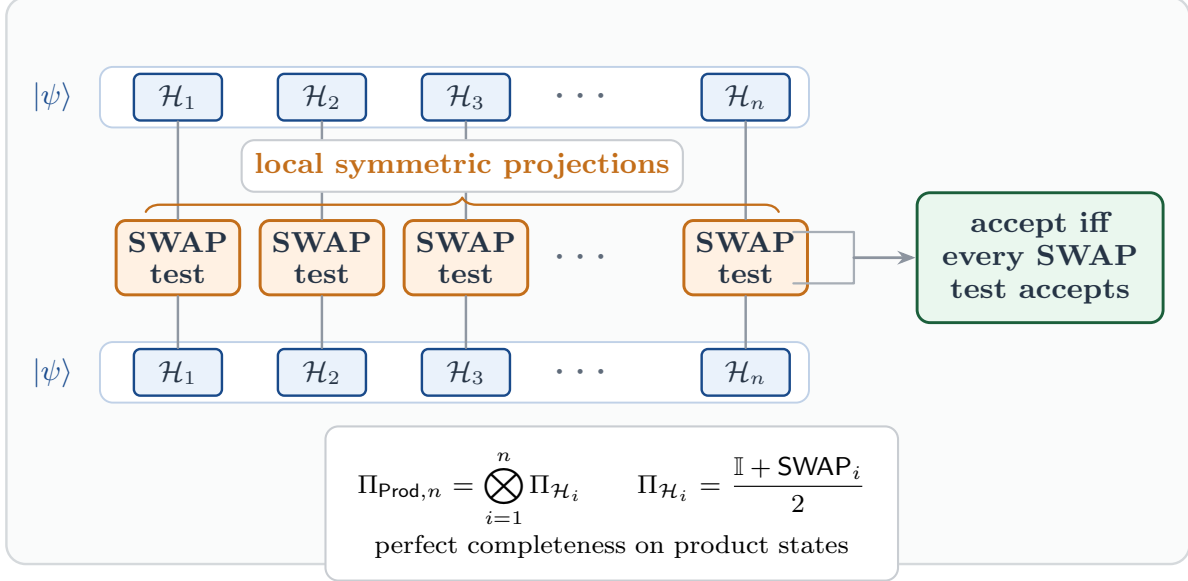

Harrow and Montanaro placed this test at the center of the theory of unentangled quantum proofs and tensor optimization \cite{harrow2010Efficient}.  A key consequence of their analysis is the structural collapse \(\QMA(k)=\QMA(2)\) for every \(k\ge2\), showing that a constant number of unentangled quantum witnesses can be reduced to two, see~\cite{JWL26} for a survey on $\QMA(2)$. Later, Soleimanifar and Wright revisited the product test before extending it to test properties of matrix product states \cite{soleimanifar2022Testing}. Their inductive argument gave a simpler proof of the soundness of the product test together with the exact high-overlap branch of the curve. Still, the low-overlap regime was left open, leaving the full soundness curve undetermined. This work determines that curve, i.e., how well the product test performs as a function of the distance of the input from product states.

\subsection{The exact soundness problem}
Define the closest-product-state overlap of $\ket{\psi}$ as
\begin{align}\label{eq:overlap-intro}
  \ov_n(\ket{\psi})=
  \max_{\ket{v_i}\in\mathcal H_i}
  \abs{\braket{\psi}{v_1\otimes\cdots\otimes v_n}}^2.
\end{align}
The trace distance between $\ket\psi$ and the set of product states is thus $\sqrt{1-\ov_n(\ket{\psi})}$.  Writing $\PT_n(\ket{\psi})$ for the acceptance probability of the product test, we study the worst-case acceptance probability
\begin{equation}\label{eq:pt-curve-intro}
  \PT_n(\omega)
  =\sup\set{\PT_n(\ket{\psi}):\ov_n(\ket{\psi})=\omega},
\end{equation}
where the supremum ranges over all finite local dimensions.  In words, \(\PT_n(\omega)\) asks for the hardest state to reject among all states whose best product approximation has fidelity exactly \(\omega\).

Harrow and Montanaro proved that if \(\omega=1-\epsilon\), then the rejection probability is \(\Theta(\epsilon)\), uniformly over the number of parties and local dimensions \cite{harrow2010Efficient}.  Soleimanifar and Wright proved the sharper upper bound
\begin{equation}\label{eq:sw-bound-intro}
  \PT_n(\omega)
  \le \min\left\{1-\omega+\omega^2,
                   \frac23+\frac13\omega^2\right\},
\end{equation}
which is tight for \(\omega\ge1/2\) \cite{soleimanifar2022Testing}.  Hence the large-overlap part of the curve was already known:
\begin{equation}\label{eq:known-large-regime}
  \PT_n(\omega)=1-\omega+\omega^2
  \qquad (\omega\ge1/2).
\end{equation}
The unresolved regime was \(\omega<1/2\), where the state is far from every product state. The $d$-dimensional maximally entangled state $\ket{\psi}=\frac{1}{\sqrt{d}} \sum_{i=1}^d \ket{ii}$ has \(\omega=1/d\) and product-test acceptance \(\frac12(1+1/d)\), suggesting that the limiting worst-case acceptance probability is $1/2$. But the known envelopes did not determine the curve between the reciprocal points or below \(1/2\).  Determining the behavior of the test in this small-$\omega$ (large-$\epsilon$) regime, including the limiting value of the worst-case acceptance probability, was posed as an open problem by Harrow and Montanaro~\cite[Section~8, Problem~2]{harrow2010Efficient}, and Soleimanifar and Wright asked explicitly whether \(\PT_n(\omega)\to1/2\) as \(\omega\to0\)~\cite{soleimanifar2022Testing}. In the special case 
of bipartite states, $n=2$, Lovitz and Lowe \cite{lovitz2026Nearly} answer both questions and they ask for the general multipartite case, $n > 2$ with arbitrary $\omega < 1/2$. Our main result, Theorem~\ref{thm:main-intro}, resolves both questions in the general case.

\subsection{Main result}
The exact soundness curve is governed by a classical capped-simplex problem, which can be stated as follows. Among all probability vectors whose entries are at most \(\omega\), maximize the collision probability:
\begin{equation}\label{eq:s-def-intro}
  s(\omega)=
  \max\left\{\sum_j p_j^2:
    p_j\ge0,
    \sum_jp_j=1,
    p_j\le\omega\ \text{for all }j\right\}.
\end{equation}
If we let
\begin{equation}\label{eq:m-r-intro}
  m=\left\lfloor\frac1\omega\right\rfloor
  \quad \text{and} \quad
  r=1-m\omega,
\end{equation}
then the maximizing vector is obtained by letting $p_j=\omega$ for exactly $m$ entries, putting the residual mass \(r\) on one final entry, and setting the remaining $p_j$'s equal to zero.  Consequently, the maximum in Eq.~\ref{eq:s-def-intro} can be expressed in closed form as
\begin{equation}\label{eq:s-closed-intro}
  s(\omega)=m\omega^2+r^2.
\end{equation}

\begin{theorem}[Exact all-regimes product-test curve]\label{thm:main-intro}
For every \(n\ge2\) and every \(\omega\in(0,1]\),
\begin{equation}\label{eq:main-intro}
  \PT_n(\omega)=\frac{1+s(\omega)}2.
\end{equation}
Equivalently, with \(m=\lfloor1/\omega\rfloor\),
\begin{equation}\label{eq:main-piecewise-intro}
  \PT_n(\omega)=\frac12\left(1+m\omega^2+(1-m\omega)^2\right).
\end{equation}
On the interval \(\omega\in(1/(m+1),1/m]\), this is the quadratic
\begin{equation}\label{eq:main-piece-m-intro}
  \PT_n(\omega)=1-m\omega+\frac{m(m+1)}2\omega^2.
\end{equation}
Moreover, the supremum is achieved by a bipartite state tensored with product states on any additional registers.
\end{theorem}

In Fig.~\ref{fig:comparison}, we plot our curve along with those of Refs.~\cite{harrow2010Efficient,soleimanifar2022Testing} for comparison. For \(m=1\), Theorem~\ref{thm:main-intro} recovers the tight Soleimanifar--Wright branch \(1-\omega+\omega^2\)~\cite{soleimanifar2022Testing}.  The first new pieces are \(1-2\omega+3\omega^2\) on \(1/3<\omega\le1/2\), \(1-3\omega+6\omega^2\) on \(1/4<\omega\le1/3\), and \(1-4\omega+10\omega^2\) on \(1/5<\omega\le1/4\).  The pieces meet continuously at the reciprocal breakpoints.  Since every feasible capped distribution satisfies \(\sum_jp_j^2\le\omega\), the theorem immediately implies the low-overlap limit.

\begin{corollary}[Small-overlap limit]\label{cor:small-overlap}
For every \(n\ge2\),
\[
  \lim_{\omega\to0}\PT_n(\omega)=\frac12.
\]
More quantitatively,
\[
  \frac12\le \PT_n(\omega)\le \frac12+\frac\omega2.
\]
\end{corollary}

\paragraph{Comparison with previous upper bounds.}
The prior all-regimes upper bound of Soleimanifar and Wright~\cite{soleimanifar2022Testing} is the envelope
\[
  B_{\mathrm{SW}}(\omega)
  =\min\left\{\omega^2-\omega+1,\ \frac13\omega^2+\frac23\right\}.
\]
The Harrow--Montanaro bound from \cite{harrow2010Efficient}, written in the same \(\omega=1-\epsilon\) parameterization and as plotted in \cite[Figure~2]{soleimanifar2022Testing}, is
\[
  B_{\mathrm{HM}}(\omega)
  =\min\left\{\omega+(1-\omega)^2+(1-\omega)^{3/2},\
  1-\frac{11}{512}(1-\omega)\right\}.
\]
Figure~\ref{fig:comparison} plots these two envelopes together with the exact curve proved here.  The exact curve agrees with \(\omega^2-\omega+1\) throughout \([1/2,1]\).  Below \(1/2\), it separates from the Soleimanifar--Wright envelope and tends to \(1/2\), whereas that envelope tends to \(2/3\) and the Harrow--Montanaro envelope tends to \(1-11/512\).

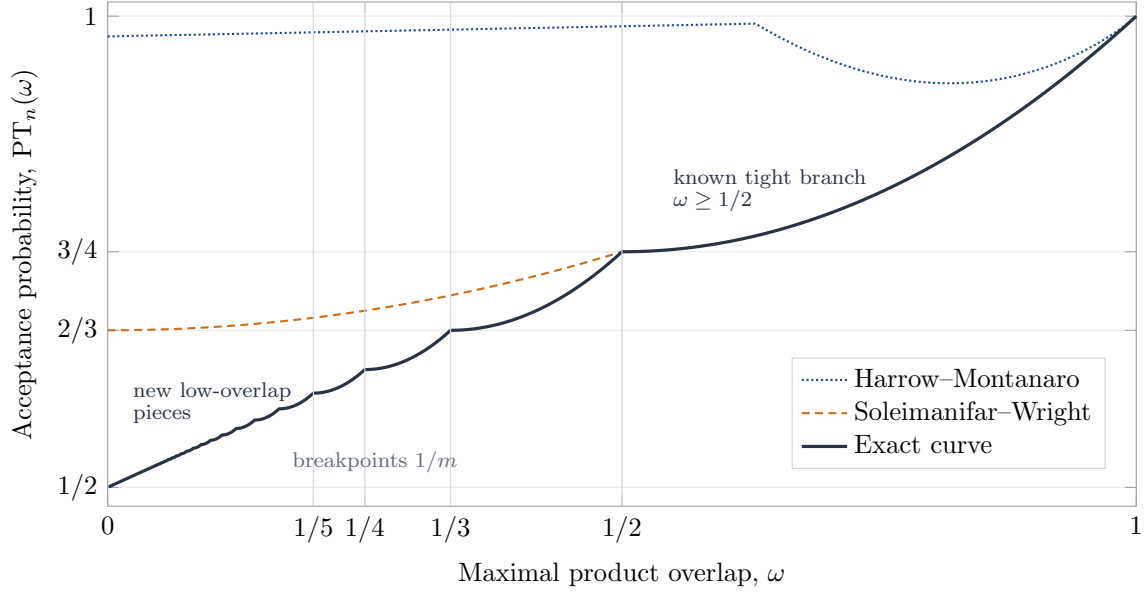
\begin{figure}[t]
\centering
\begin{tikzpicture}
\begin{axis}[
  width=0.92\linewidth,
  height=0.50\linewidth,
  xmin=0, xmax=1,
  ymin=0.48, ymax=1.015,
  xlabel={Maximal product overlap, $\omega$},
  ylabel={Acceptance probability, $\PT_n(\omega)$},
  xtick={0,0.2,0.25,0.3333333,0.5,1},
  xticklabels={$0$,$1/5$,$1/4$,$1/3$,$1/2$,$1$},
  ytick={0.5,0.6666667,0.75,1},
  yticklabels={$1/2$,$2/3$,$3/4$,$1$},
  axis line style={draw=ptgray!65},
  tick style={draw=ptgray!65},
  grid=major,
  major grid style={draw=ptgray!15},
  legend style={at={(0.665,0.295)},anchor=north west,draw=ptgray!30,fill=white,fill opacity=0.96,text opacity=1,font=\small},
  legend cell align=left,
  tick label style={font=\small},
  label style={font=\small},
]
\addplot[ptblue!86!black, thick, densely dotted, domain=0:1, samples=500, no marks] {min(x+(1-x)^2+(1-x)^1.5, 1 - (11/512)*(1-x))};
\addlegendentry{Harrow--Montanaro}
\addplot[ptorange!92!black, thick, densely dashed, domain=0:1, samples=500, no marks] {min(x^2-x+1, x^2/3+2/3)};
\addlegendentry{Soleimanifar--Wright}
\addplot[ptdark, very thick, domain=0.001:1, samples=1800, no marks] {exactpt(x)};
\addlegendentry{Exact curve}

\draw[ptgray!24, line width=.45pt] (axis cs:.5,0.48) -- (axis cs:.5,1.015);
\draw[ptgray!24, line width=.45pt] (axis cs:.333333333,0.48) -- (axis cs:.333333333,1.015);
\draw[ptgray!24, line width=.45pt] (axis cs:.25,0.48) -- (axis cs:.25,1.015);
\draw[ptgray!24, line width=.45pt] (axis cs:.2,0.48) -- (axis cs:.2,1.015);
\node[anchor=west, font=\scriptsize, text=ptdark, align=left] at (axis cs:.54,.815)
  {known tight branch\\[-.4mm]$\omega\ge 1/2$};
\node[anchor=west, font=\scriptsize, text=ptdark, align=left] at (axis cs:.015,.59)
  {new low-overlap\\[-.4mm]pieces};
\node[anchor=south, font=\scriptsize, text=ptgray] at (axis cs:.26,.505)
  {breakpoints $1/m$};
\end{axis}
\end{tikzpicture}
\caption{\textbf{Exact product-test acceptance curve for every n.} It shows that the curve of the bipartite case, $n=2$, of Lovitz and Lowe \cite{lovitz2026Nearly} is extremal for every $n \ge 2$. Compared with the two previous envelopes from \cite[Figure~2]{soleimanifar2022Testing}. The exact curve is piecewise quadratic with breakpoints at $1,1/2,1/3,\ldots$, agrees with the known branch on $[1/2,1]$, and tends to $1/2$ in the low-overlap limit.}
\label{fig:comparison}
\end{figure}

The theorem also gives the exact property-testing soundness statement in trace distance: if \(\ket\psi\) is \(\delta\)-far from every product state, the product test accepts with probability at most \((1+s(1-\delta^2))/2\), which equals \(1-\delta^2+\delta^4\) when \(\delta\le1/\sqrt2\); see Section~\ref{sec:consequences}.

\subsection{Why the capped curve is the right curve}

The lower bound is bipartite.  If \(\lambda_j\) are the squared Schmidt coefficients of a bipartite pure state, then
\[
  \PT_2(\ket{\psi})=\frac12\left(1+\sum_j \lambda_j^2\right),
  \qquad
  \ov_2(\ket{\psi})=\max_j \lambda_j.
\]
Thus the bipartite extremal problem is exactly the capped collision problem: maximize \(\sum_j \lambda_j^2\) subject to \(\lambda_j\le\omega\).  Theorem~\ref{thm:main-intro} says that this bipartite case is already the worst possible one.  Extra parties do not create a state that is harder for the product test to reject.

The upper bound is the main point.  Cut the state between the first register and the remaining registers,
\[
  \ket\psi=
  \sum_i\sqrt{\lambda_i}\ket{a_i}\ket{b_i}.
\]
After the first local SWAP test, the diagonal branches \(\ket{b_i}^{\otimes2}\) have weights \(\lambda_i^2\), while the off-diagonal branches have total weight \(\sum_{i<j}\lambda_i\lambda_j\).  A sharp first-SWAP reduction keeps every diagonal branch; combined with induction on the number of parties, it gives
\[
  \PT_n(\ket{\psi})
  \le
  \frac12+\frac12\sum_i\lambda_i^2s(\phi_i),
  \qquad
  \phi_i=\ov_{n-1}(\ket{b_i}).
\]
The product-overlap constraint transfers to each branch as \(\lambda_i\phi_i\le\omega\).  Refining branch \(i\) by a capped optimizer for \(s(\phi_i)\) produces one global probability distribution whose entries all have mass at most \(\omega\).  Its second moment is therefore at most \(s(\omega)\), which completes the induction.

This is the only new accounting compared with the high-overlap analyses.  When \(\omega\ge1/2\), an extremal capped distribution has at most two entries, so a two-branch analysis is sufficient (see Sec.~\ref{sec:capped}).  When \(\omega<1/2\), several diagonal branches can matter simultaneously.  Keeping all of them is what exposes the full piecewise-quadratic curve.

\subsection{Related work}

The product test belongs to the broader program of quantum property testing, surveyed by Montanaro and de Wolf \cite{montanaro2016Survey}.  In state testing, one asks for a structural promise about an unknown state using as few copies as possible; examples include spectrum and identity testing \cite{odonnell2015Quantum,yu2021Sample}, sequentially defined finite properties \cite{harrow2017Sequential}, and tests for multipartite entanglement structure.  Several recent papers separate nearby productness models that are easy to conflate.  Jones and Montanaro study testing whether a state is product across some bipartition rather than fully product across all named registers \cite{jones2025Testing}.  Bouland, Giurgica-Tiron, and Wright study a hidden-cut version where the product bipartition is unknown \cite{bouland2025State}.  Beckey, Coffman, Shlosberg, Schatzki, and Leditzky study the restrictions imposed by single-copy measurements \cite{beckey2025Product}.  The present paper keeps the measurement model fixed--the canonical joint two-copy local symmetric measurement--and determines its exact worst-case behavior.

The tensor-network viewpoint gives another motivation.  Matrix product states and tree tensor network states are central models of low-entanglement structure, with algorithmic foundations tied to area laws and one-dimensional Hamiltonians \cite{hastings2007Area,landau2015Polynomial}.  Soleimanifar and Wright showed that product states are an exceptional constant-copy endpoint: once the bond dimension is at least two, testing matrix product structure requires copy complexity that grows with the number of qudits \cite{soleimanifar2022Testing}.  Lovitz and Lowe extend this picture to tree tensor networks \cite{lovitz2026Nearly}, and Chen, Wang, and Zhang prove local-test lower bounds for bipartite unitarily invariant properties, with consequences for Schmidt-rank and MPS testing \cite{chen2024Local}.  Our result gives the exact soundness profile at this exceptional endpoint.

Finally, the parameter \(\ov_n(\ket{\psi})\) is the closest-product-state fidelity; under the common logarithmic convention, \(-\log \ov_n(\ket{\psi})\) is the geometric measure of entanglement \cite{wei2003Geometric}.  It is also the optimization objective in the closest-product-state learning problem studied by Bakshi, Bostanci, Kretschmer, Landau, Li, Liu, O'Donnell, and Tang \cite{bakshi2025Learning}.  That line of work asks how well one can find or estimate the closest product state.  Here the closest product state is used as a promise parameter: our main theorem gives the exact response of the canonical two-copy test as a function of that fidelity.

\subsection{Organization}

Section~\ref{sec:prelim} fixes notation and records basic product-test identities.  Section~\ref{sec:capped} proves the capped-simplex lemmas.  Section~\ref{sec:lower} gives the matching lower-bound construction.  Section~\ref{sec:first-swap} proves the sharp first-SWAP reduction.  Section~\ref{sec:upper} proves the upper bound and Theorem~\ref{thm:main-intro}.  Section~\ref{sec:qma-collapse} records the sharpened parameter in the Harrow--Montanaro reduction from $\QMA(k)$ to $\QMA(2)$.  Section~\ref{sec:consequences} gives trace-distance and finite-dimensional consequences, along with the analogy with Razborov's triangle-density curve.

\section{Preliminaries}\label{sec:prelim}

All Hilbert spaces are finite-dimensional over \(\C\), and all pure states are unit vectors unless stated otherwise.  The local dimensions are arbitrary throughout; the bounds below do not depend on them.  For \(n\ge1\), let $\mathcal H_{[n]}\coloneqq\mathcal H_1\otimes\cdots\otimes\mathcal H_n$. We next define the primitive upon which the product test is built.

\begin{definition}[SWAP test]\label{def:swap-test}
For a Hilbert space \(\mathcal H\), let
\begin{align}
     \Pi_{\mathcal H}=\frac{\mathbb{I}+\SWAP_{\mathcal H}}2
\end{align}
be the projector onto the symmetric subspace of \(\mathcal H\otimes\mathcal H\).  Given two states \(\ket\phi,\ket{\phi'}\in\mathcal H\), the \emph{SWAP test} is the two-outcome measurement \(\{\Pi_{\mathcal H},\Id-\Pi_{\mathcal H}\}\) on \(\ket\phi\otimes\ket{\phi'}\), accepting on the symmetric outcome.  Its acceptance probability is
\[
  \norm{\Pi_{\mathcal H}\ket\phi\ket{\phi'}}^2=\tfrac12\bigl(1+\abs{\langle{\phi|\phi'\rangle}}^2\bigr),
\]
so the test accepts with certainty when \(\ket\phi=\ket{\phi'}\) and with probability \(1/2\) when they are orthogonal.
\end{definition}

\begin{definition}[Product test]\label{def:product-test}
Given two copies of a state \(\ket\psi\in\mathcal H_{[n]}=\mathcal H_1\otimes\cdots\otimes\mathcal H_n\), the \emph{product test} performs a SWAP test on each pair of corresponding local registers in parallel, accepting only if all \(n\) tests accept.  Its accepting projector is
\[
  \Pi_{\Prod,n}=\Pi_{\mathcal H_1}\otimes\cdots\otimes\Pi_{\mathcal H_n},
\]
where \(\Pi_{\mathcal H_i}\) acts on the two copies of the \(i\)-th local register, and its acceptance probability is
\[
  \PT_n(\ket{\psi})=\norm{\Pi_{\Prod,n}\ket\psi^{\otimes2}}^2.
\]
For \(n=1\), the product test is a single SWAP test on two identical copies, hence \(\PT_1(\ket{\psi})=1\).
\end{definition}

\begin{lemma}[{Purity formula for the product test; cf.~\cite[Lemma~2]{harrow2010Efficient}}]\label{lem:purity-formula}
For \(\ket\psi\in\mathcal H_{[n]}\), let \(\rho_S\) denote the reduced density matrix of \(\ket\psi\bra\psi\) on the tensor factors indexed by \(S\subseteq[n]\), with \(\rho_\emptyset\) interpreted as the scalar \(1\).  Then
\[
  \PT_n(\ket{\psi})=\frac1{2^n}\sum_{S\subseteq[n]}\Tr(\rho_S^2).
\]
\end{lemma}
\begin{proof}
Expanding the tensor product of local symmetric subsystem projectors as a sum over all choices of \(\mathbb{I}_i\) versus \(\SWAP_i\) in each factor gives a binomial-like expansion indexed by subsets \(S\subseteq[n]\):
\begin{align}
    \Pi_{\Prod,n}
  =\frac1{2^n}\bigotimes_{i=1}^n(\mathbb{I}_i+\SWAP_i)
  =\frac1{2^n}\sum_{S\subseteq[n]}\bigotimes_{i\in S}\SWAP_i\otimes\bigotimes_{i\notin S}\mathbb{I}_i
  =\frac1{2^n}\sum_{S\subseteq[n]}\SWAP_S,
\end{align}
  
where in the last step we have defined
\[
  \SWAP_S=\prod_{i\in S}\SWAP_i,
\]
to be the operator that swaps the two copies on every register in \(S\) and acts as the identity on registers outside \(S\).  Because the local swaps act on disjoint tensor factors and therefore commute, this product is well-defined and equals the tensor product written above.

We now apply the standard ``swap trick.''  For any pure state \(\ket\psi\in\mathcal H_{[n]}\) and any \(S\subseteq[n]\),
\[
  \bra\psi^{\otimes2}\SWAP_S\ket\psi^{\otimes2}
  =\Tr\bigl(\SWAP_S(\ket\psi\bra\psi\otimes\ket\psi\bra\psi)\bigr)
  =\Tr\bigl(\SWAP_S\,(\rho_S\otimes\rho_S)\bigr)
  =\Tr(\rho_S^2),
\]
where the second equality uses that \(\SWAP_S\) acts trivially on the registers outside \(S\), so tracing those registers out on each copy independently produces \(\rho_S\otimes\rho_S\), and the third equality is the standard swap trick \(\Tr(\SWAP\,(A\otimes B))=\Tr(AB)\) applied on the registers in \(S\).  Since \(\Pi_{\Prod,n}\) is an orthogonal projector, \(\PT_n(\ket{\psi})=\norm{\Pi_{\Prod,n}\ket\psi^{\otimes2}}^2=\bra\psi^{\otimes2}\Pi_{\Prod,n}\ket\psi^{\otimes2}\); summing over \(S\) and using the expansion of \(\Pi_{\Prod,n}\) then yields the desired formula. Note that the \(S=\emptyset\) term contributes \(\Tr(\rho_\emptyset^2)=1\), consistent with the convention \(\rho_\emptyset=1\).
\end{proof}

\begin{definition}[Product overlap]\label{def:overlap}
For \(\ket\psi\in\mathcal H_{[n]}\), define
\[
  \ov_n(\ket{\psi})=
  \max_{\ket{v_i}\in\mathcal H_i}
  \abs{\langle \psi|v_1\otimes\cdots\otimes v_n\rangle}^2.
\]
The maximum is attained, since the set of product unit vectors is compact in finite dimensions and the objective is continuous.  For \(n=1\), every state is product, so \(\ov_1(\ket{\psi})=1\).
\end{definition}
Next, let us introduce a notation to represent the worst-case acceptance probability for states having overlap $\omega$ with the nearest product state.
\begin{definition}[Dimension-free worst cases]\label{def:worst-case}
For \(n\ge2\) and \(\omega\in(0,1]\), define
\[
  \PT_n(\omega)
  =\sup\left\{\PT_n(\ket{\psi}):\ov_n(\ket{\psi})=\omega
  \right\}.
\]
We also use the monotone variant
\[
  \PT_n^{\le}(\omega)
  =\sup\left\{\PT_n(\ket{\psi}): \ov_n(\ket{\psi})\le\omega\right\},
\]
where the supremum is again over all finite local dimensions.
\end{definition}
As we will see, the worst case will be obtained by a bipartite state. Thus, the following lemma will be essential; it appears as \cite[Lemma~20]{harrow2010Efficient}, where it is noted to be implicit in earlier work of Wei and Goldbart~\cite{wei2003Geometric}, and we include a short proof to keep the paper self-contained.

\begin{lemma}[Bipartite overlap and acceptance]\label{lem:bipartite}
Let \(\ket\psi\in\mathcal{H}_A\otimes\mathcal{H}_B\) have Schmidt decomposition
\[
  \ket\psi=\sum_j\sqrt{\lambda_j}\ket{x_j}\ket{y_j},
  \qquad \lambda_j>0,
  \qquad \sum_j \lambda_j=1.
\]
Then
\[
  \ov_2(\ket{\psi})=\max_j \lambda_j
\]
and
\[
  \PT_2(\ket{\psi})=\frac{1}{2}\left(1+\sum_j \lambda_j^2\right).
\]
\end{lemma}

\begin{proof}
For vectors $\ket a\in\mathcal{H}_A, \ket b\in\mathcal{H}_B$, write $\alpha_j=\langle x_j|a \rangle $ and $\beta_j=\langle y_j | b\rangle$, so that \(\sum_j\abs{\alpha_j}^2\le1\) and \(\sum_j\abs{\beta_j}^2\le1\).  If \(\lambda_{\max}\coloneqq\max_j \lambda_j\), then
\begin{align}
     \abs{\langle \psi | a \otimes b \rangle}
  =\abs{\sum_j\sqrt{\lambda_j}\,\alpha_j\beta_j}
  \le \sqrt{\lambda_{\max}}
     \left(\sum_j\abs{\alpha_j}^2\right)^{1/2}
     \left(\sum_j\abs{\beta_j}^2\right)^{1/2}
  \le \sqrt{\lambda_{\max}}.
\end{align}
Squaring both sides and noting that equality is obtained by taking \(\ket a=\ket{x_j}\) and \(\ket b=\ket{y_j}\) for an index \(j\) attaining \(\lambda_{\max}\), yields $\ov_2(\ket{\psi})=\lambda_{\max}$. Using Lemma~\ref{lem:purity-formula}, and noting that the eigenvalues of $\rho_A$ (and $\rho_B$) are the numbers $\lambda_j$, we obtain
\[
  \PT_2(\ket{\psi})=\frac14\left(2+2\Tr(\rho_A^2)\right)
  =\frac12\left(1+\sum_j \lambda_j^2\right).\qedhere
\]
\end{proof}

\section{The capped second moment}\label{sec:capped}
In this section, we isolate the elementary classical optimization behind the exact curve. We note that this calculation only becomes relevant in the low-overlap regime. This occurs because, when the cap is below $1/2$, no feasible distribution can concentrate all mass on two entries. Thus, an extremizer must have more than two nonzero entries. As we decrease $\omega$, it necessitates more non-zero entries in any feasible distribution, leading to the piecewise nature of the final curve.

\begin{definition}[Capped collision probability]\label{def:s}
For \(\omega\in(0,1]\), let
\[
  s(\omega)=
  \max\left\{\sum_jp_j^2:\ p_j\ge0,\ \sum_jp_j=1,\ p_j\le\omega\ \text{for every }j\right\},
\]
where the optimization ranges over finite probability vectors; the maximum is attained, as the next lemma shows.
\end{definition}

\begin{figure}[H]
\centering
\resizebox{0.91\linewidth}{!}{%
\begin{tikzpicture}[
  >=Stealth,
  line cap=round,
  line join=round,
  every node/.style={font=\small},
  panel/.style={draw=ptgray!26, fill=ptgray!3, rounded corners=4pt, line width=.75pt},
  axis/.style={draw=ptgray!72, line width=.76pt},
  capline/.style={draw=ptorange!95!black, densely dashed, line width=.95pt},
  satbar/.style={draw=ptblue!86!black, fill=ptlightblue, rounded corners=1.8pt, line width=.80pt},
  rembar/.style={draw=ptpurple!84!black, fill=ptlightpurple, rounded corners=1.8pt, line width=.80pt},
  ghost/.style={draw=ptgray!30, fill=ptgray!7, rounded corners=1.8pt, line width=.55pt},
  call/.style={draw=ptgray!35, fill=white, rounded corners=4pt, align=center,
    inner sep=6pt, line width=.68pt, font=\footnotesize, text=ptdark},
  tag/.style={font=\footnotesize\bfseries},
  note/.style={font=\footnotesize, text=ptgray},
  arrow/.style={-{Stealth[length=1.8mm]}, draw=ptgray!62, line width=.76pt}
]
\draw[panel] (-.62,-.5) rectangle (11.52,3.6);

\draw[axis] (.70,.60) -- (6.60,.60);
\draw[axis] (.70,.60) -- (.70,2.94);
\node[rotate=90, anchor=south, note] at (.22,1.70) {mass};
\node[note, anchor=north] at (3.62,.15) {coordinates};
\draw[capline] (.55,2.34) -- (6.45,2.34);
\node[anchor=east, tag, text=ptorange!92!black] at (.6,2.34) {cap $\omega$};

\foreach \x/\lab in {1.10/1,1.78/2,2.46/3,3.14/4} {
  \draw[satbar] (\x,.60) rectangle +(0.44,1.74);
  \node[note, anchor=north] at (\x+.22,.45) {$\lab$};
}
\node[font=\Large, text=ptgray] at (4.12,1.39) {$\cdots$};
\draw[satbar] (5.03,.60) rectangle +(0.44,1.74);
\node[note, anchor=north] at (5.25,.45) {$m$};
\draw[rembar] (5.93,.60) rectangle +(0.44,.86);
\node[note, anchor=north] at (6.15,.45) {};
\node[tag, text=ptpurple!84!black] at (6.15,1.67) {$r$};
\draw[arrow, draw=ptpurple!70!black] (6.15,1.55) -- (6.15,1.38);

\draw[decorate, decoration={brace, amplitude=4.3pt}, draw=ptblue!82!black, line width=.72pt]
  (1.02,2.74) -- (5.56,2.74)
  node[midway, yshift=.34cm, tag, text=ptblue!84!black]
  {$m=\lfloor1/\omega\rfloor$ saturated coordinates};

\draw[draw=ptgray!22, line width=.6pt] (6.92,.34) -- (6.92,3.26);
\node[call, text width=3.88cm] at (9.22,1.72)
  {greedy optimizer\\[-.2mm]
   $\displaystyle p^\star=(\omega,\ldots,\omega,r)$\\[.25mm]
   $\displaystyle r=1-m\omega$\\[3mm]
   collision probability\\[.3mm]
   $\displaystyle s(\omega)=m\omega^2+r^2$\\[.35mm]
   };
\end{tikzpicture}%
}
\caption{\textbf{The capped-simplex optimizer.}  Under the cap $p_j\le\omega$, collision probability is maximized by concentrating mass as much as the cap permits: fill $m=\lfloor1/\omega\rfloor$ entries to height $\omega$, then put the residual mass $r=1-m\omega$ on one final entry.}
\label{fig:capped-spectrum}
\end{figure}
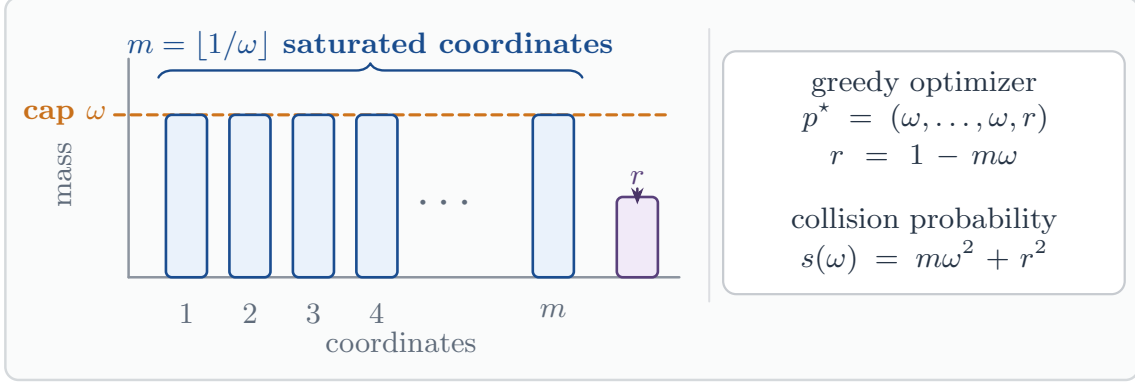

\begin{lemma}[Capped-simplex optimizer]\label{lem:capped-simplex}
Let \(\omega\in(0,1]\), set \(m=\lfloor1/\omega\rfloor\), and set \(r=1-m\omega\).  Then
\[
  s(\omega)=m\omega^2+r^2.
\]
The value is achieved by the vector \((\omega,\ldots,\omega,r)\), with \(m\) copies of \(\omega\) and with the final coordinate omitted if \(r=0\).  In particular, \(s\) is nondecreasing and \(s(\omega)\le\omega\).
\end{lemma}

\begin{proof}
Let \(p=(p_1,\ldots,p_d)\) be any feasible finite probability vector in $\mathbb{R}^d$.  We first show that, without decreasing \(\sum_jp_j^2\), it can be transformed into a vector with at most one coordinate strictly between \(0\) and \(\omega\), all other coordinates being either \(0\) or \(\omega\). To this end, suppose two coordinates satisfy \(0<p_i\le p_j<\omega\).  Then, let
\begin{align}
     \delta=\min\{p_i,\omega-p_j\}>0,
\end{align}
and replace \((p_i,p_j)\) by \((p_i-\delta,p_j+\delta)\).  Feasibility is preserved, and the change in the collision probability is
\[
  (p_i-\delta)^2+(p_j+\delta)^2-p_i^2-p_j^2
  =2\delta(p_j-p_i)+2\delta^2>0.
\]
Each such operation either sets one of the two chosen coordinates to \(0\) or saturates the other at \(\omega\).  Hence the number of positive unsaturated coordinates decreases by at least one, and the procedure terminates after finitely many steps.  The resulting feasible vector has collision probability at least that of \(p\) and has at most one nonzero unsaturated coordinate.

After this transformation, the resulting vector has \(q\) coordinates equal to \(\omega\), at most one further coordinate \(t\in[0,\omega)\), and all remaining coordinates equal to zero.  The normalization constraint gives
\[
  q\omega+t=1.
\]
Since \(t<\omega\), this forces \(q=\lfloor1/\omega\rfloor=m\) and \(t=1-m\omega=r\).  Therefore every feasible vector has collision probability at most \(m\omega^2+r^2\), and the displayed capped vector achieves this value.

Monotonicity follows directly from the definition: increasing the cap only enlarges the feasible set.  Finally, \(s(\omega)\le\omega\) because every feasible vector satisfies \(p_j^2\le\omega p_j\), thus \(\sum_jp_j^2\le\omega\sum_jp_j=\omega\).
\end{proof}

\begin{remark}[Why a classical lemma appears]
For a bipartite state the quantities \(p_j\) are squared Schmidt coefficients.  In the induction below they are instead refined branch weights produced by the first SWAP test.  The same capped second-moment problem controls both situations, which is why the final curve is independent of the number of parties.
\end{remark}

The following refinement lemma is a key step in the inductive proof below.  It says that subdividing masses cannot beat the best global cap.

\begin{lemma}[Capped refinement inequality]\label{lem:refinement}
Let \(w_1,\ldots,w_d\ge0\) satisfy \(\sum_i w_i=1\). Let \(\alpha_i\in(0,1]\) be such that
\[
  w_i\alpha_i\le\omega
\]
for every \(i\) with \(w_i>0\). Then
\[
  \sum_i w_i^2s(\alpha_i)\le s(\omega).
\]
\end{lemma}

\begin{proof}
Indices with \(w_i=0\) contribute nothing and may be ignored. For each remaining \(i\), choose a finite probability vector \(q^{(i)}=(q^{(i)}_1,q^{(i)}_2,\dots)\) attaining \(s(\alpha_i)\). Thus \(q^{(i)}_j\le\alpha_i\) for every \(j\), \(\sum_jq^{(i)}_j=1\), and \(\sum_j(q^{(i)}_j)^2=s(\alpha_i)\). Next, define a refined probability vector \(v\) indexed by pairs \((i,j)\) as \(v_{ij}=w_iq^{(i)}_j\). Then \(\sum_{i,j}v_{ij}=1\), and
\[
  v_{ij}=w_iq^{(i)}_j\le w_i\alpha_i\le\omega.
\]
Therefore \(v\) is feasible for \(s(\omega)\), and
\[
  s(\omega)
  \ge \sum_{i,j}v_{ij}^2
  =\sum_i w_i^2\sum_j(q^{(i)}_j)^2
  =\sum_i w_i^2s(\alpha_i).\qedhere
\]
\end{proof}

\section{The lower bound}\label{sec:lower}
To prove our main result, Theorem~\ref{thm:main-intro}, we first prove that the worst-case acceptance probability is at least $(1+s(\omega))/2$. This construction provides important intuition for the complementary upper-bound analysis. It is attained by an explicit bipartite state whose squared Schmidt coefficients are the capped optimizer from Section~\ref{sec:capped}.  Thus any possible multipartite improvement in the upper bound would have to beat an already sharp bipartite obstruction.

\begin{proposition}[Extremizing construction]\label{prop:lower}
For every \(n\ge2\) and every \(\omega\in(0,1]\),
\[
  \PT_n(\omega)\ge\frac{1+s(\omega)}2.
\]
Moreover, the same lower bound holds for \(\PT_n^{\le}(\omega)\).
\end{proposition}

\begin{proof}
Let \(m=\lfloor1/\omega\rfloor\) and \(r=1-m\omega\).  On \(\C^{m+1}\otimes\C^{m+1}\), with the final basis vector omitted if \(r=0\), consider the state
\begin{equation}\label{eq:extremizer}
  \ket{\psi_\omega}
  =\sum_{j=1}^m\sqrt\omega\ket j\ket j
   +\sqrt r\ket{m+1}\ket{m+1},
\end{equation}
with $m$ squared Schmidt coefficients of $\omega$ and, if \(r>0\), one copy of $r<\omega$.  By Lemma~\ref{lem:bipartite},
\[
  \ov_2(\ket{\psi_\omega})=\omega
\]
and
\[
  \PT_2(\psi_\omega)
  =\frac12\left(1+m\omega^2+r^2\right)
  =\frac{1+s(\omega)}2.
\]
For \(n>2\), tensor \(\ket{\psi_\omega}\) with arbitrary product states on registers \(3,\ldots,n\).  The product overlap is unchanged, since the extra product registers contribute overlap at most one, a bound that can be saturated. The product-test acceptance probability is also unchanged, since the additional local SWAP tests accept identical pure states with probability one.  Hence the construction proves the claimed lower bounds for both exact overlap and overlap at most \(\omega\).
\end{proof}

\section{A sharp first-SWAP reduction}\label{sec:first-swap}

This section contains the only genuinely quantum step in the proof.  The inductive proof due to Soleimanifar and Wright~\cite{soleimanifar2022Testing} proceeds by taking a Schmidt decomposition between the first and remaining registers. They keep only the largest diagonal contribution and ``charitably bound the probability by $1$'' in all remaining branches, diagonal and off-diagonal alike.  The key insight leading to the exact curve is that one must avoid discarding the smaller diagonal branches: below \(\omega=1/2\), the extremal capped spectrum necessarily has three or more nonzero entries. The reduction below keeps all diagonal branches and pays only for the off-diagonal branches by the operator norm of the remaining projector. It turns out that this is sufficient to fully determine the desired worst-case acceptance probability.

Let \(\ket\psi\in\mathcal H_1\otimes\mathcal{H}_{[2:n]}\), where $\mathcal{H}_{[2:n]}\coloneqq \mathcal H_2\otimes\cdots\otimes\mathcal H_n$.  Write its Schmidt decomposition across the cut \(\mathcal H_1\mid\mathcal{H}_{[2:n]}\) as
\begin{align}\label{eq:first-cut-schmidt}
  \ket\psi=\sum_i\sqrt{\lambda_i}\ket{a_i}\ket{b_i},
\end{align}
where \(\lambda_i>0\), \(\sum_i\lambda_i=1\), and the \(\ket{a_i}\)'s and \(\ket{b_i}\)'s are orthonormal families in $\mathcal{H}_1$ and $\mathcal{H}_{[2:n]}$, respectively.

\begin{lemma}[First-SWAP reduction]\label{lem:first-swap}
For \(n\ge2\), in the setting of \eqref{eq:first-cut-schmidt},
\begin{equation}\label{eq:first-swap-reduction}
  \PT_n(\ket{\psi})
  \le
  \sum_i\lambda_i^2\PT_{n-1}(\ket{b_i})
  +\sum_{i<j}\lambda_i\lambda_j.
\end{equation}
\end{lemma}

\begin{proof}
Let \(\Pi_1\) be the SWAP-test accepting projector on the two copies of \(\mathcal H_1\), and let
\[
  \Pi_{\mathrm{rest}}=\Pi_2 \otimes\cdots\otimes\Pi_n
\]
be the product-test accepting projector on the last \(n-1\) pairs of local registers.  Thus \(\Pi_{\mathrm{rest}}\) is not the full symmetric projector on \(\mathcal{H}_{[2:n]}\otimes\mathcal{H}_{[2:n]}\), but rather the tensor product of the local symmetric projectors.  
Expanding two copies of \eqref{eq:first-cut-schmidt} and applying \(\Pi_1\) gives
\begin{align}\label{eq:after-first-swap}
  \Pi_1\ket\psi^{\otimes2}
  &=
  \sum_i\lambda_i\ket{a_i}^{\otimes2}\ket{b_i}^{\otimes2}
  +
  \sum_{i<j}\sqrt{\lambda_i\lambda_j}
  \left(\frac{\ket{a_i}\ket{a_j}+\ket{a_j}\ket{a_i}}{\sqrt2}\right)
  \left(\frac{\ket{b_i}\ket{b_j}+\ket{b_j}\ket{b_i}}{\sqrt2}\right).
\end{align}
The first-register vectors
\[
  \ket{a_i}^{\otimes2}
  \quad\text{and}\quad
  \frac{\ket{a_i}\ket{a_j}+\ket{a_j}\ket{a_i}}{\sqrt2}\quad(i<j)
\]
are pairwise orthonormal.  Since \(\Pi_{\mathrm{rest}}\) acts only on the remaining registers, applying it preserves orthogonality between the branches indexed by \(i\) and by \(i<j\).  Therefore
\begin{align}
  \PT_n(\ket{\psi})
  &=\norm{\Pi_{\mathrm{rest}}\Pi_1\ket\psi^{\otimes2}}^2 \\
  &=\sum_i\lambda_i^2
    \norm{\Pi_{\mathrm{rest}}\ket{b_i}^{\otimes2}}^2
    +
    \sum_{i<j}\lambda_i\lambda_j
    \norm{\Pi_{\mathrm{rest}}
      \frac{\ket{b_i}\ket{b_j}+\ket{b_j}\ket{b_i}}{\sqrt2}}^2 \\
  &\le
    \sum_i\lambda_i^2\PT_{n-1}(\ket{b_i})
    +\sum_{i<j}\lambda_i\lambda_j,
\end{align}
where we have used the fact that $\Pi_{\mathrm{rest}}\leq \mathbb{I}$ and the normalization of the state.
\end{proof}

\begin{remark}
The inequality in Lemma~\ref{lem:first-swap} is the only lossy quantum estimate in the argument: the off-diagonal branches are merely bounded by the operator norm of the remaining projector.  The theorem is nevertheless sharp because the extremal examples are bipartite and have no genuinely multipartite qualities to exploit.
\end{remark}

\begin{lemma}[Branch-overlap constraint]\label{lem:branch-overlap}
In the setting of \eqref{eq:first-cut-schmidt}, let
\[
  \omega=\ov_n(\ket{\psi}),
  \qquad
  \phi_i=\ov_{n-1}(\ket{b_i}).
\]
Then, for every $i$, we have
\[
  \lambda_i\phi_i\le\omega.
\]
\end{lemma}

\begin{proof}
Let $\ket{v_i}=\ket{v_{i,2}}\otimes\cdots\otimes\ket{v_{i,n}}\in\mathcal{H}_{[2:n]}$ be a product state attaining $\abs{\langle b_i|v_i\rangle}^2=\phi_i$. The vector $\ket{a_i}\otimes\ket{v_i}$ is a product state on all $n$ registers. Since the $\ket{a_j}$'s are orthonormal, expanding $\ket\psi=\sum_j\sqrt{\lambda_j}\,\ket{a_j}\otimes\ket{b_j}$ gives
\[
  \langle\psi|\bigl(\ket{a_i}\otimes\ket{v_i}\bigr)
  =\sum_j\sqrt{\lambda_j}\,\langle a_j|a_i\rangle\,\langle b_j|v_i\rangle
  =\sqrt{\lambda_i}\,\langle b_i|v_i\rangle,
\]
so
\[
  \bigl|\langle\psi|\bigl(\ket{a_i}\otimes\ket{v_i}\bigr)\bigr|^2
  =\lambda_i\abs{\langle b_i|v_i\rangle}^2
  =\lambda_i\phi_i.
\]
By definition, $\omega=\ov_n(\ket{\psi})$ is the maximum of $\abs{\langle\psi|v\rangle}^2$ over product states $\ket v$. Since $\ket{a_i}\otimes\ket{v_i}$ is one such product state,
\[
  \lambda_i\phi_i
  =\bigl|\langle\psi|\bigl(\ket{a_i}\otimes\ket{v_i}\bigr)\bigr|^2
  \le\ov_n(\ket{\psi})=\omega,
\]
as claimed.
\end{proof}

\section{The upper bound}\label{sec:upper}

We now prove the matching upper bound by induction on the number of parties. Lemma~\ref{lem:first-swap} converts the product-test probability into diagonal branch contributions plus off-diagonal mass.  Lemma~\ref{lem:branch-overlap} transfers the global overlap promise to each branch.  Lemma~\ref{lem:refinement} then turns the branch data into one capped probability vector and applies the definition of \(s(\omega)\).

\begin{proposition}[Upper bound on $\PT_n(\ket{\psi})$]\label{prop:upper}
For every \(n\ge1\), every finite-dimensional \(n\)-partite pure state \(\ket\psi\), and \(\omega=\ov_n(\ket{\psi})\),
\[
  \PT_n(\ket{\psi})\le\frac{1+s(\omega)}2.
\]
\end{proposition}

\begin{proof}
We prove the stated pointwise inequality for all \(n\ge1\) by induction on \(n\); for \(n=1\) we use the convention \(\ov_1(\ket{\psi})=1\), so the only possible value of \(\omega\) is \(1\).  If \(n=1\), then \(\ov_1(\ket{\psi})=1\), \(\PT_1(\psi)=1\), and \((1+s(1))/2=1\).

Assume the statement for \((n-1)\)-partite states, and let \(\ket\psi\in\mathcal H_1\otimes\cdots\otimes\mathcal H_n\).  Write the Schmidt decomposition across \(\mathcal H_1\mid(\mathcal H_2\otimes\cdots\otimes\mathcal H_n)\) as in \eqref{eq:first-cut-schmidt}.  Let
\[
  \omega=\ov_n(\ket{\psi}),
  \qquad
  \phi_i=\ov_{n-1}(\ket{b_i}).
\]
We may then write
\begin{align*}
  \PT_n(\ket{\psi})
  &\le
  \sum_i\lambda_i^2\PT_{n-1}(b_i)
  +\sum_{i<j}\lambda_i\lambda_j, \quad \text{Lemma}~\ref{lem:first-swap} \\
  &\le
  \sum_i\lambda_i^2\frac{1+s(\phi_i)}2
  +\sum_{i<j}\lambda_i\lambda_j, \quad \text{induction hypothesis},\\
  &= \left( \frac{1}{2} \sum_i \lambda_i^2 + \sum_{i < j} \lambda_i \lambda_j \right)+ \frac{1}{2}\sum_i \lambda_i^2 s(\phi_i).
\end{align*}
Using $\sum_i\lambda_i=1$, one may verify that
\begin{align*}
      \frac12\sum_i\lambda_i^2+
  \sum_{i<j}\lambda_i\lambda_j
  =\frac12\left(\sum_i\lambda_i\right)^2
  =\frac12,
\end{align*}
which yields the upper bound
\begin{equation}\label{eq:upper-before-refinement}
  \PT_n(\ket{\psi})
  \le
  \frac12+\frac12\sum_i\lambda_i^2s(\phi_i).
\end{equation}
By Lemma~\ref{lem:branch-overlap}, \(\lambda_i\phi_i\le\omega\) for every \(i\).  This is the exact hypothesis needed to apply the capped refinement inequality, Lemma~\ref{lem:refinement}.  Doing so yields
\begin{align*}
    \sum_i\lambda_i^2s(\phi_i)\le s(\omega).
\end{align*}
Substituting this into \eqref{eq:upper-before-refinement} proves
\[
  \PT_n(\ket{\psi})\le\frac12+\frac12s(\omega)=\frac{1+s(\omega)}2.\qedhere
\]
\end{proof}

We have now proved all of the essential ingredients in the proof of our main theorem. 
\begin{proof}[Proof of Theorem~\ref{thm:main-intro}]
The lower bound is Proposition~\ref{prop:lower}, and the upper bound is Proposition~\ref{prop:upper}.  The closed forms \eqref{eq:main-piecewise-intro} and \eqref{eq:main-piece-m-intro} follow from Lemma~\ref{lem:capped-simplex}.
\end{proof}

The same proof also gives the monotone version used in property-testing soundness statements.

\begin{corollary}[At-most-overlap version]\label{cor:at-most}
For every \(n\ge2\) and every \(\omega\in(0,1]\),
\[
  \PT_n^{\le}(\omega)=\frac{1+s(\omega)}2.
\]
Equivalently, every state satisfying \(\ov_n(\ket{\psi})\le\omega\) satisfies
\[
  \PT_n(\ket{\psi})\le\frac{1+s(\omega)}2.
\]
\end{corollary}

\begin{proof}
If \(\ov_n(\ket{\psi})=\omega'\le\omega\), then Proposition~\ref{prop:upper} gives
\[
  \PT_n(\ket{\psi})\le\frac{1+s(\omega')}2\le\frac{1+s(\omega)}2,
\]
where the last step uses monotonicity of \(s\).  Proposition~\ref{prop:lower} supplies a state of exact overlap \(\omega\) attaining the displayed value.
\end{proof}

\section{Improved parameter in the Harrow--Montanaro collapse}\label{sec:qma-collapse}

Harrow and Montanaro use the product test to convert a $k$-Merlin protocol into a two-Merlin protocol~\cite{harrow2010Efficient}.  This section records the numerical improvement obtained by substituting the exact product-test curve for the coarse all-regimes estimate in their Lemma~5; the separability statement recorded below corresponds to their Lemma~6.  Following their notation, let $\QMA_m(k)_{\sigma,c}$ denote $k$-Merlin protocols with $m$-qubit messages, soundness $\sigma$, and completeness $c$.

The point of the exact analysis needed here is the following global rejection bound.  If a $k$-partite state has closest-product overlap at most $1-\eps$, with $0\le\eps<1$, then Corollary~\ref{cor:at-most} gives
\begin{equation}\label{eq:hm-rejection-exact}
  1-\PT_k(\psi)
  \ge
  1-\frac{1+s(1-\eps)}2
  =\frac{1-s(1-\eps)}2.
\end{equation}
Since $s(\omega)\le\omega$, this implies the sharp uniform linear estimate
\begin{equation}\label{eq:hm-rejection-linear}
  1-\PT_k(\psi)\ge \frac{\eps}{2}.
\end{equation}
The constant $1/2$ here is optimal: on each piece $\omega\in(1/(m+1),1/m]$ one has $\omega-s(\omega)=(1-m\omega)\bigl((m+1)\omega-1\bigr)$, which vanishes precisely at the reciprocal points $\omega=1/d$, where the $d$-dimensional maximally entangled state has $\eps=1-1/d$ and is rejected with probability exactly $\eps/2$.  This replaces the Harrow--Montanaro all-regimes estimate $1-\PT_k(\psi)\ge(11/512)\eps$ in the two-Merlin simulation.

\begin{proposition}[Improved two-Merlin simulation parameter]\label{prop:hm-improved-parameter}
For every $m$, every $k\ge2$, and every $0\le \sigma<c\le1$, the Harrow--Montanaro simulation protocol has completeness at least
\[
  c^\star=\frac{1+c}{2}
\]
and soundness at most
\[
  \sigma^\star=1-\frac{(1-\sigma)^2}{4}.
\]
Its accepting operator is separable across the two Merlins.  Consequently, whenever $c^\star>\sigma^\star$,
\begin{equation}\label{eq:hm-improved-inclusion}
  \QMA_m(k)_{\sigma,c}
  \subseteq
  \QMA_{km}(2)_{\sigma^\star,c^\star}.
\end{equation}
In applications to the collapse, one first amplifies the original protocol so that this final promise gap is positive.
\end{proposition}

\begin{proof}
The verifier is the Harrow--Montanaro verifier: each of the two Merlins sends a $k$-block state, and Arthur chooses uniformly between the product test on the two received $k$-block states and the original $k$-Merlin verification procedure applied to one of the two received states.  Completeness is unchanged: honest Merlins send two copies of an optimal product witness, the product test accepts with certainty, and the verification branch accepts with probability at least $c$, so the completeness is at least $c^\star=(1+c)/2$.

Now consider a no instance.  The total acceptance probability is convex in each of the two unentangled Merlin messages separately.  Therefore, when maximizing over two-Merlin strategies $\rho_1\otimes\rho_2$, it suffices to maximize over extreme points of the density-matrix sets.  Thus it suffices to consider pure $k$-block messages $\ket{\phi_1}$ and $\ket{\phi_2}$.  Write
\[
  \ov_k(\ket{\phi_i})=1-\eps_i\qquad(i=1,2).
\]
Let $1-\Delta_i$ be the acceptance probability of the product test on two identical copies of $\ket{\phi_i}$.  By \eqref{eq:hm-rejection-linear},
\begin{equation}\label{eq:hm-self-reject}
  \Delta_i\ge\frac{\eps_i}{2}.
\end{equation}

We next compare identical-copy and two-input product-test acceptance probabilities.  For pure $k$-partite states $\ket\phi$ and $\ket\chi$, let
\[
  \PT_k(\phi,\chi)=\bra{\phi\otimes\chi}\Pi_{\Prod,k}\ket{\phi\otimes\chi}
\]
denote the acceptance probability of the product test executed on the input $\ket\phi\otimes\ket\chi$, so that $\PT_k(\phi,\phi)=\PT_k(\phi)$, and write $\rho_S^\phi$ and $\rho_S^\chi$ for the reduced density matrices on the registers indexed by $S\subseteq[k]$.  Expanding the accepting projector as in the proof of Lemma~\ref{lem:purity-formula}, which is the two-state form of \cite[Lemma~2]{harrow2010Efficient}, gives
\begin{equation}\label{eq:cross-product-test-formula}
  \PT_k(\phi,\chi)
  =2^{-k}\sum_{S\subseteq[k]} \Tr(\rho_S^\phi\rho_S^\chi).
\end{equation}
By Cauchy--Schwarz in the Hilbert--Schmidt direct sum over all subsets $S$,
\begin{equation}\label{eq:hm-cross-cs}
  \PT_k(\phi,\chi)
  \le
  \sqrt{\PT_k(\phi,\phi)\PT_k(\chi,\chi)}
  \le
  \frac{\PT_k(\phi,\phi)+\PT_k(\chi,\chi)}2.
\end{equation}
Let $1-\Delta$ be the acceptance probability of the product test when its two inputs are $\ket{\phi_1}$ and $\ket{\phi_2}$.  Applying \eqref{eq:hm-cross-cs} to these two states gives
\[
  1-\Delta
  \le \frac{(1-\Delta_1)+(1-\Delta_2)}2.
\]
Hence, if
\[
  \bar\eps=\frac{\eps_1+\eps_2}{2},
\]
then \eqref{eq:hm-self-reject} implies
\begin{equation}\label{eq:hm-cross-reject}
  \Delta\ge\frac{\Delta_1+\Delta_2}{2}
  \ge \frac{\eps_1+\eps_2}{4}
  =\frac{\bar\eps}{2}.
\end{equation}

On the verification branch, choose a closest product state to each $\ket{\phi_i}$.  The original $k$-Merlin soundness is $\sigma$ on product witnesses, and the usual pure-state continuity estimate for any $0\le M\le I$ (see, e.g., \cite[Lemma~22]{harrow2010Efficient}) gives an additive loss at most $\sqrt{\eps_i}$ on input $\ket{\phi_i}$.  Since Arthur chooses one of the two received states uniformly, the verification branch accepts with probability at most
\begin{equation}\label{eq:hm-verifier-branch}
  \min\left\{1,\ \sigma+\frac{\sqrt{\eps_1}+\sqrt{\eps_2}}2\right\}
  \le
  \min\left\{1,\ \sigma+\sqrt{\bar\eps}\right\},
\end{equation}
where the last inequality is concavity of the square root.

Combining the two equally likely branches, the total acceptance probability is at most
\begin{equation}\label{eq:hm-soundness-optimization}
  \frac12\left(1-\frac{\bar\eps}{2}
  +\min\{1,\sigma+\sqrt{\bar\eps}\}\right).
\end{equation}
Set $a=1-\sigma$.  If $\bar\eps\ge a^2$, then \eqref{eq:hm-soundness-optimization} is at most
\[
  \frac12\left(2-\frac{\bar\eps}{2}\right)
  \le 1-\frac{a^2}{4}.
\]
If $\bar\eps\le a^2$, then $\sigma+\sqrt{\bar\eps}\le1$, and the right-hand side of \eqref{eq:hm-soundness-optimization} becomes
\[
  1-\frac a2+\frac{\sqrt{\bar\eps}}2-\frac{\bar\eps}{4}.
\]
The function $u\mapsto \sqrt u/2-u/4$ is increasing on $[0,1]$, so this is maximized over $0\le\bar\eps\le a^2$ at $\bar\eps=a^2$, where it equals $1-a^2/4$.  Thus the soundness is at most $\sigma^\star=1-(1-\sigma)^2/4$.

It remains only to record separability of the accepting measurement, following \cite[Lemma~6]{harrow2010Efficient}.  For a local register $\mathcal H$ of dimension $d$,
\[
  \Pi_{\mathcal H}
  =\frac{d(d+1)}2\int \ket{x}\bra{x}\otimes\ket{x}\bra{x}\,dx,
\]
where the integral is with respect to Haar measure on unit vectors.  Hence each local symmetric projector is separable across the two copies, and their tensor product, the product-test accepting operator, is separable across the two Merlins.  The verification branch has accepting operators of the form $M\otimes I$ and $I\otimes M$, which are product positive operators across the two-Merlin cut.  Convex combinations preserve separability, so the full accepting operator lies in $\mathsf{SEP}$.
\end{proof}

\begin{remark}[Effect on the collapse constants]
The equality $\QMA(k)=\QMA(2)$ is unchanged, but the one-shot soundness loss in the reduction is better.  The converted protocol has completeness at least $c^\star=(1+c)/2$ and soundness at most $\sigma^\star=1-(1-\sigma)^2/4$; as usual, the starting protocol is amplified before this conversion when these one-shot parameters do not already give a positive final gap.  Harrow and Montanaro's displayed parameter $1-(1-\sigma)^2/100$ is replaced by $1-(1-\sigma)^2/4$.  In particular, if a $k$-Merlin protocol has soundness $\sigma=1-1/q$, then the converted two-Merlin protocol has soundness at most
\[
  1-\frac{1}{4q^2}.
\]
In the non-perfect-completeness branch of their amplification argument, the preliminary amplification step \cite[Lemma~8]{harrow2010Efficient} turns a promise gap of $c-\sigma=1/q$ into soundness $1-1/(3q)$, and the conversion then yields soundness at most
\[
  1-\frac{1}{36q^2}.
\]
Thus the product-test part of the collapse pays a quadratic loss with constant $4$ rather than $100$.
\end{remark}

\begin{remark}[Keeping the full exact curve]
The proof above uses only the clean consequence $s(\omega)\le\omega$ of the exact curve.  One can retain the full curve in the same calculation.  For $0\le\eps<1$, define
\[
  R(\eps)=1-\PT_k^{\le}(1-\eps)=\frac{1-s(1-\eps)}2.
\]
Extend this continuously to $\eps=1$ by setting $s(0)=0$, so $R(1)=1/2$; equivalently, the endpoint $\eps=1$ may be interpreted as a limiting value.  Then the same proof bounds the soundness $\sigma_{\mathrm{ex}}$ of the converted protocol by the sharper, but less transparent, expression
\[
  \sigma_{\mathrm{ex}}\le
  \max_{0\le\eps_1,\eps_2\le1}
  \frac12\left(
    1-\frac{R(\eps_1)+R(\eps_2)}2
    +\min\left\{1,\sigma+\frac{\sqrt{\eps_1}+\sqrt{\eps_2}}2\right\}
  \right).
\]
The closed form in Proposition~\ref{prop:hm-improved-parameter} is the uniform parameter obtained from the optimal global linear rejection constant $1/2$.
\end{remark}

\section{Consequences and discussion}\label{sec:consequences}

\paragraph{Trace-distance soundness.}
For pure states, the trace distance between \(\ket\psi\) and \(\ket\phi\) is
\[
  \frac12\bigl\|\ket\psi\bra\psi-\ket\phi\bra\phi\bigr\|_1
  =\sqrt{1-\abs{\langle \psi|\phi \rangle}^2}.
\]
Consequently, the distance from \(\ket\psi\) to the set of product states is \(\sqrt{1-\ov_n(\ket{\psi})}\).  Corollary~\ref{cor:at-most} therefore gives the following property-testing form: if \(\ket\psi\) is \(\delta\)-far in trace distance from every product state, then
\[
  \PT_n(\ket{\psi})
  \le \frac{1+s(1-\delta^2)}2.
\]
When \(\delta\le1/\sqrt2\), the right-hand side is \(1-\delta^2+\delta^4\), so the product test rejects with probability at least \(\delta^2(1-\delta^2)\).  For arbitrary constant \(\delta>0\), the exact formula still gives a constant rejection probability independent of \(n\) and of the local dimensions.

\paragraph{Fixed dimensions.}
Theorem~\ref{thm:main-intro} is dimension-free: the upper bound applies to every finite choice of local dimensions, while the matching lower bound is allowed to choose dimensions large enough to realize the capped Schmidt spectrum. If the local dimensions are fixed in advance, the same upper bound remains valid verbatim. The only possible change is on the lower-bound side, where the capped spectrum may not fit in the available dimensions. That finite-dimensional extremal problem is separate from the dimension-independent soundness curve determined here.

\paragraph{Independence of the number of parties.}
For every \(n\ge2\), the dimension-free worst-case curve is the same.  Additional parties cannot increase the extremal acceptance probability: a worst case may always be realized by a bipartite state tensored with arbitrary product states.  This is a dimension-free statement.  The upper bound is valid for every prescribed finite collection of local dimensions, while the matching lower-bound construction may require two local dimensions of size about \(1/\omega\).

\paragraph{Explicit extremizing family.}
For \(\omega\in(1/(m+1),1/m]\), one extremizing family has squared Schmidt coefficients
\[
  \omega,\ldots,\omega,1-m\omega.
\]
At reciprocal points \(\omega=1/d\), this family can be chosen as the \(d\)-dimensional maximally entangled state, and
\[
  \PT_n(1/d)=\frac12\left(1+\frac1d\right).
\]
Between reciprocal points, the extremizing distribution is not uniform: it is the capped distribution with as many coordinates of mass \(\omega\) as possible and one residual coordinate.

\paragraph{Resolution of the low-overlap regime.}
By Corollary~\ref{cor:small-overlap},
\[
  \frac12\le \PT_n(\omega)\le\frac12+\frac\omega2,
\]
so \(\PT_n(\omega)\to1/2\) as \(\omega\to0\).  The theorem gives more than the limit: it identifies the complete fine structure of the convergence as the capped collision probability of the most concentrated probability distribution with weights at most \(\omega\).

\begin{remark}[A piecewise analogy with Razborov's triangle-density curve]\label{rem:razborov}
The shape in Figure~\ref{fig:comparison} is structurally reminiscent of Razborov's exact minimum triangle-density curve for graphs of fixed edge density \cite{razborov2008Minimal}.  In Razborov's theorem the formula changes at the Tur\'an densities \(1-1/k\), reflecting when the extremal multipartite graphon activates another part.  Here the formula changes at the reciprocal caps \(1/\omega\), reflecting when the capped Schmidt distribution activates another weight.  The analogy is only structural -- there is no reduction between the problems -- but in both settings an exact extremal profile is assembled from simple analytic pieces indexed by the number of active parts.
\end{remark}

\paragraph{Why the low-overlap regime needs a new accounting.}
The proof differs from the induction of Soleimanifar and Wright at one point.  After the first local SWAP test, their argument keeps the dominant diagonal Schmidt branch and upper-bounds the remaining contribution.  This is tight when \(\omega\ge1/2\), because the extremizing bipartite construction can have only two nonzero squared Schmidt coefficients.  For \(\omega<1/2\), however, no probability distribution whose probabilities are at most \(\omega\) can put all its mass on two coordinates.  Several diagonal branches can contribute nontrivially.  Retaining all diagonal branches changes the residual optimization from a two-branch calculation into the capped refinement inequality of Lemma~\ref{lem:refinement}, producing exactly the additional quadratic pieces.

\subsection*{Acknowledgments}

ChatGPT 5.5 Pro was used as a research assistant to help explore the analysis of the product test. ChatGPT 5.5 Pro and Claude Fable 5 were also used to assist in the writing, while the authors rewrote, revised, and improved the exposition and citations. In addition to human verification, Codex 5.5 was used to assist with an exploratory auto-formalization of the improved product-test analysis in Lean. The authors take full responsibility for the correctness, exposition, and attribution in the final manuscript.

\bibliographystyle{alpha}
{\small \bibliography{main}}

\end{document}